%% file: main.tex
\documentclass[11pt]{article}

\input{preamble}

\title{The Born Rule as the Unique Refinement-Stable Induced Weight on Robust Record Sectors}
\author{Marko Lela\\\href{https://orcid.org/0009-0008-0768-5184}{ORCID: 0009-0008-0768-5184}}
\date{March 24, 2026}

\begin{document}

\maketitle

\begin{abstract}
This paper proves a conditional structural uniqueness theorem for induced weight on robust record sectors within an admissible Hilbert record layer. Its theorem target and additive carrier differ from those of the standard Born-rule routes: additivity is not placed on the full projector lattice, but on disjoint admissible continuation bundles through an extensive bundle valuation, from which the sector-level additive law is inherited under admissible refinement. Accordingly, the result is not a Gleason-type representation theorem in different language, but a distinct uniqueness theorem about induced sector weight inherited from bundle additivity on admissible continuation structure. Under two explicit structural conditions, internal equivalence of admissible binary refinement profiles and sufficient admissible refinement richness, the quadratic assignment is the only non-negative refinement-stable induced weight on robust record sectors. In the main theorem, refinement richness is secured by admissible binary saturation. A supplementary proposition shows that dense admissible saturation already suffices if continuity of the profile function is added. Under normalization, the result reduces to the standard Born assignment.
\end{abstract}

\input{sections/01_introduction}
\input{sections/02_framework}
\input{sections/03_profile_reduction}
\input{sections/04_structural_conditions}
\input{sections/05_main_theorem}
\input{sections/06_scope_and_limits}
\input{sections/07_relation_to_existing_routes}

\input{sections/08_conclusion}

\bibliographystyle{plain}
\bibliography{bib/references}

\end{document}

%% file: preamble.tex
\usepackage[T1]{fontenc}
\usepackage[utf8]{inputenc}
\usepackage{lmodern}
\usepackage[margin=1in]{geometry}
\usepackage{microtype}

\usepackage{amsmath,amssymb,amsthm,mathtools}
\usepackage{bm}
\usepackage{enumitem}
\usepackage{csquotes}
\usepackage{array}

\usepackage{xcolor}
\definecolor{linkblue}{RGB}{0,70,140}

\usepackage[
  colorlinks=true,
  linkcolor=linkblue,
  citecolor=linkblue,
  urlcolor=linkblue
]{hyperref}
\usepackage[nameinlink,noabbrev]{cleveref}

\newtheorem{definition}{Definition}
\newtheorem{lemma}{Lemma}
\newtheorem{proposition}{Proposition}
\newtheorem{theorem}{Theorem}
\newtheorem{corollary}{Corollary}
\theoremstyle{remark}
\newtheorem{remark}{Remark}

\newtheorem{condition}[theorem]{Condition}
\numberwithin{equation}{section}

\crefname{definition}{definition}{definitions}
\Crefname{definition}{Definition}{Definitions}

\crefname{lemma}{lemma}{lemmas}
\Crefname{lemma}{Lemma}{Lemmas}

\crefname{proposition}{proposition}{propositions}
\Crefname{proposition}{Proposition}{Propositions}

\crefname{theorem}{theorem}{theorems}
\Crefname{theorem}{Theorem}{Theorems}

\crefname{corollary}{corollary}{corollaries}
\Crefname{corollary}{Corollary}{Corollaries}

\crefname{remark}{remark}{remarks}
\Crefname{remark}{Remark}{Remarks}

\crefname{condition}{condition}{conditions}
\Crefname{condition}{Condition}{Conditions}

\newcommand{\RR}{\mathbb{R}}
\newcommand{\HH}{\mathcal{H}}
\newcommand{\Aexp}{\mathcal{A}_{\mathrm{exp}}}
\newcommand{\Cpsi}[1]{C_{\Psi}\!\left(#1\right)}
\newcommand{\Wpsi}[1]{W_{\Psi}\!\left(#1\right)}
\newcommand{\Proj}[1]{\Pi_{#1}}
\newcommand{\norm}[1]{\left\lVert #1 \right\rVert}

\hypersetup{
  pdftitle={The Born Rule as the Unique Refinement-Stable Induced Weight on Robust Record Sectors},
  pdfauthor={Marko Lela},
  pdfsubject={Quantum foundations},
  pdfkeywords={Born rule, quantum foundations, record sectors, refinement stability, induced weight, continuation bundles}
}

%% file: sections/01_introduction.tex
\section{Introduction}
\label{sec:introduction}

Why the quadratic quantum weight is singled out remains a central question in the foundations of quantum theory. Most standard routes to the Born rule begin by fixing a broad target and then proving that the quadratic assignment is the unique object compatible with that target. In Gleason-type approaches, the target is a measure on the full lattice of projections or, more generally, on generalized observables \cite{Gleason1957,Busch2003,Caves2002}. In Everettian decision-theoretic approaches, the target is rational preference or rational credence in branching settings \cite{Deutsch1999,Wallace2010,Wallace2012}. In envariance-based approaches, the target is a probability assignment extracted from symmetry properties of entangled states \cite{Zurek2003,Zurek2005}. Other routes appeal to self-locating uncertainty, evidential coherence, operational reconstruction, or normative Bayesian constraints \cite{SebensCarroll2014,GreavesMyrvold2010,MasanesGalleyMuller2019,FuchsSchack2013}. The present paper therefore starts from a different theorem target, a different additive carrier, and hence a different uniqueness problem.

The present paper asks a different question. It does not begin with a measure on all projectors, with a rationality postulate, or with a symmetry principle on entangled states. Instead, it asks which non-negative weight assignments are structurally admissible on systems that can function as robust internal records. The target is therefore narrower but also sharper: induced weights on robust record sectors inside an admissible Hilbert record layer.

This shift in target changes the logical location of the crucial additive step. Rather than postulating additivity on the full projector lattice, the paper places the additive primitive on disjoint admissible continuation bundles through an extensive bundle valuation. The sector-level additive law is then inherited from continuation partition under admissible refinement. The resulting theorem is therefore not a global measure-representation theorem on projectors, but a conditional structural uniqueness theorem for induced weights on robust record sectors. The difference from standard routes is structural, not merely verbal: the object being weighted is different, the structure carrying additivity is different, and the uniqueness problem is posed at a different level.

The second key move is to impose an internal equivalence condition at the level of admissible binary refinement profiles. This does not yet assume that induced weight is a function of norm. It says only that record situations with the same admissible binary refinement profile must carry the same induced weight. The reduction to a one-variable function of the norm of the projected record component is obtained only after admissible binary saturation is added, because saturation makes the admissible binary refinement profile fully classifiable by that norm. Under this additional structural condition, the problem reduces to a function
\[
g : \RR_{\geq 0} \to \RR_{\geq 0}
\]
of the norm of the projected component. The remaining question is then whether the admissible refinement structure is rich enough to force the functional equation
\[
g\!\left(\sqrt{r_1^2+r_2^2}\right) = g(r_1) + g(r_2)
\qquad
\text{for all } r_1,r_2 \in \RR_{\geq 0}.
\]
In the main theorem, this refinement richness is secured by admissible binary saturation. A supplementary proposition shows that dense admissible saturation already suffices if continuity of the profile function is added. Non-negativity then forces the additive function to be linear, by \cref{lem:cauchy}. The quadratic assignment follows directly.

The resulting theorem is conditional and should be read as such, but its conditional form is part of its content rather than a defect. It neither derives Hilbert structure nor assigns weights to arbitrary projectors, and it does not claim that every orthogonal decomposition is physically meaningful or that a global probability calculus has been reconstructed from no assumptions. What it does show is sharper and more limited: it identifies a precise structural threshold within an admissible Hilbert record layer at which the quadratic assignment becomes the only non-negative refinement-stable induced weight on robust record sectors.

This conditional form is a methodological virtue. Foundational arguments often become difficult to assess because the assumptions doing the actual work are dispersed across the proof. The present route instead isolates them and turns them into an explicit threshold statement. One structural condition concerns internal equivalence at the level of admissible binary refinement profiles. The other concerns sufficient admissible refinement richness to force the functional equation, realized in the main theorem by admissible binary saturation and weakened in a supplementary proposition to dense admissible saturation plus continuity. Whether one accepts or rejects the conclusion can therefore be tied to clearly identifiable commitments rather than to hidden background assumptions.

The main theorem may be summarized informally as follows. Let $\Psi$
be a global representational state and let $R$ be a robust record
sector in an admissible Hilbert record layer. If induced weight is
fixed by internally accessible admissible refinement structure, and if
the admissible refinement structure is rich enough to force all
norm-compatible binary profiles, then the quadratic assignment is the
only non-negative refinement-stable induced weight on robust record
sectors. Equivalently, there exists a constant $c \geq 0$ such that
\[
\Wpsi{R} = c\,\norm{\Proj{R}\Psi}^2.
\]
Under normalization, this reduces to the standard Born assignment
$\Wpsi{R} = \norm{\Proj{R}\Psi}^2$.

The novelty of the present route therefore lies less in the final mathematical form than in the theorem target, the additive carrier, and the structural level at which the uniqueness problem is posed. This paper does not replace Gleason's theorem. It proves a different uniqueness theorem about a different object under a different additive carrier. Gleason concerns global additive assignments on the projector lattice. The present paper concerns induced weights on robust record sectors, with the additive primitive placed on disjoint continuation bundles. Once one focuses on robust record sectors and on the continuation structure that makes them internally stable, the quadratic assignment is forced to be the only non-negative refinement-stable induced weight compatible with the stated conditions.

The paper is organized as follows. \Cref{sec:framework} introduces the admissible Hilbert record layer, robust record sectors, continuation bundles, and admissible orthogonal refinements. \Cref{sec:profile-reduction} shows how the uniqueness problem reduces to a norm-based one-variable weight function. \Cref{sec:structural-conditions} states the two structural conditions and derives the central functional equation. \Cref{sec:main-theorem} proves the quadratic uniqueness theorem. \Cref{sec:scope-limits} clarifies the precise domain and limits of the result. \Cref{sec:existing-routes} compares the present route with the major existing families of Born-rule arguments, and \Cref{sec:conclusion} closes with the main conceptual takeaway and directions for further work.

%% file: sections/02_framework.tex

\section{Framework}
\label{sec:framework}

We work conditionally within an admissible Hilbert record layer. The aim of this section is to fix the objects on which the later weight assignment is defined and to state the precise meaning of admissible refinement. No attempt is made here to derive Hilbert structure itself. That question belongs to a separate foundational program.

\subsection{Admissible record layer}

Let $\HH$ be a Hilbert space serving as an admissible Hilbert record layer for a class of internally accessible record structures. The role of this layer is representational rather than ontological: it provides a faithful coordinatization of robust record alternatives and their admissible refinements.

The basic objects in the present paper are not arbitrary projectors on $\HH$, but robust record sectors that correspond to internally stable record content.

\begin{definition}[Robust record sector]
\label{def:robust-record-sector}
A \emph{robust record sector} is a closed subspace $R \subseteq \HH$ satisfying all of the following conditions:
\begin{enumerate}[label=(\roman*)]
\item \emph{Internal discriminability}: $R$ represents a record type that is internally readable and distinguishable from incompatible admissible alternatives.
\item \emph{Short-horizon persistence}: the record content represented by $R$ is stable under admissible micro-recodings and under immediate admissible continuation, in the sense that its identifying record content does not disappear under arbitrarily small representational perturbations or one-step admissible evolution.
\item \emph{Admissible refinement closure}: $R$ belongs to an admissible exclusivity structure in which meaningful next-step alternatives can be represented by admissible orthogonal refinements of $R$.
\end{enumerate}
\end{definition}

These conditions are functional rather than microscopic. They do not require a particular dynamical mechanism in the definition itself. Their purpose is to exclude arbitrary closed subspaces from counting as record sectors merely by formal inclusion in the Hilbert layer.

\begin{remark}
Paradigmatic physical examples of robust record sectors are decoherence-stabilized pointer sectors and closely related stable record subspaces. The present paper does not build decoherence into the definition itself, because the theorem is meant to apply at the level of admissible record structure rather than at the level of any one specific stabilization mechanism.
\end{remark}

We will only consider robust record sectors that belong to the admissible layer in the above sense. In particular, not every closed subspace of $\HH$ is automatically treated as a physically meaningful record sector.

\subsection{State-relative continuation bundles}

Let $\Psi \in \HH$ denote a global representational state. Unless normalization is explicitly imposed later, representational states in the present paper are not assumed to have unit norm. The admissible Hilbert record layer is taken to be closed under positive scalar multiplication: if $\Psi$ lies in the layer and $\lambda \ge 0$, then $\lambda \Psi$ lies in the same representational layer. This closure is used only to identify the full realized norm domain before normalization is imposed. For each robust record sector $R$, we associate a state-relative continuation bundle
\[
\Cpsi{R} \subseteq \Aexp,
\]
where $\Aexp$ denotes a structured family of experientially admissible continuation bundles built from exclusive continuation alternatives within the admissible record framework used in this paper.

Informally, $\Cpsi{R}$ is the bundle of internally admissible continuations of $\Psi$ in which the record sector $R$ is realized as the relevant next record content.

We assume that this family of admissible continuation bundles carries a non-negative set function $\mu$. We will refer to $\mu$ as an \emph{extensive bundle valuation}. At this stage, the only structural property imposed on $\mu$ is finite additivity on disjoint admissible bundles.

\begin{definition}[Induced record weight]
For a global state $\Psi$ and a robust record sector $R$, the induced weight of $R$ relative to $\Psi$ is
\[
\Wpsi{R} := \mu\!\left(\Cpsi{R}\right).
\]
\end{definition}

At this stage, $\Wpsi{R}$ is only a state-relative weight on robust record sectors. It is not yet interpreted as a probability measure on all projectors.

\subsection{Extensive bundle valuation and its logical role}

The additive primitive in this paper is placed on disjoint admissible
continuation bundles. This is a positive structural choice, not a
mere reformulation of projectoral additivity.

A continuation bundle is an extensional object: it is a set of
admissible continuations, individuated by the record content they
realize. Disjoint continuation bundles are disjoint sets of
continuations. Finite additivity on such bundles is motivated directly
by the exclusivity of the continuations themselves: if no continuation
belongs to both $B_1$ and $B_2$, then the combined weight of the
two bundles is the sum of their individual weights. This is the same
motivation that underlies classical probability on mutually exclusive
events. It requires no algebraic structure on the carrier. That
exclusivity is not dispensable: if the relevant record alternatives
were allowed to overlap already at the sector level, then the
associated continuation bundles would no longer form a canonical
disjoint carrier, and additive induced weight would lose its clear
domain of application.

A projector, by contrast, is an algebraic object in the operator
algebra of a Hilbert space. Additivity on orthogonal projectors
presupposes that algebraic structure and is defined at the level of
the projector lattice. The present paper does not begin there. It
places the additive primitive on a family of exclusive continuation
bundles, and the sector-level additive law is then inherited from
admissible continuation partition. The broader comparison with
Gleason-type routes is deferred to \cref{sec:existing-routes}.

\begin{remark}[Why continuation structure is not optional for induced weight in the present framework]
In the present framework, induced weight is not introduced as a
primitive assignment on record sectors, but as the derived quantity
\(W_\Psi(R)=\mu(C_\Psi(R))\), with the additive primary carrier
located on state-relative continuation bundles. This is not a
dispensable representational choice, because it fixes the problem
type: the question is not which weight one may directly postulate on
sectors, but which sector weight is inherited from a deeper additive
structure on admissible continuations. Instantaneous collections at a
single time slice do not carry the same semantics, because
continuation bundles are not synchronic coexistence classes but
diachronic compatibility classes of forward developments relative to a
record sector. Accordingly, a modification preserves the present
problem type only if sector weight remains derived from a deeper
carrier and additivity remains anchored on that carrier rather than
being imposed directly on sectors. What the framework therefore needs
is not persistence in some maximal sense, but enough short-horizon
stability that under small time shifts or immediate admissible
evolution the classification of which forward developments still
realize the same record function does not arbitrarily change.
\end{remark}

\subsection{Projected record components}

For each robust record sector $R$, let
\[
\phi_R := \Proj{R}\Psi
\]
denote the component of $\Psi$ in $R$.

The later reduction step will show that, under explicit structural conditions, the induced weight $\Wpsi{R}$ depends only on structural information carried by $\phi_R$. In the present section, we only fix the notation.

\subsection{Admissible orthogonal refinements}

The central structural operation in the paper is not an arbitrary decomposition of a Hilbert subspace, but an admissible refinement of a robust record sector into mutually exclusive robust sub-sectors.

\begin{definition}[Admissible orthogonal refinement]
\label{def:admissible-refinement}
Let $\Psi \in \HH$ be a fixed global representational state, and let $R$ be a robust record sector. An \emph{admissible orthogonal refinement} of $R$ relative to $\Psi$ is a decomposition
\[
R = R_1 \oplus R_2
\]
such that:
\begin{enumerate}[label=(\roman*)]
\item $R_1$ and $R_2$ are themselves robust record sectors,
\item $R_1$ and $R_2$ represent mutually exclusive internally readable record alternatives,
\item the refinement faithfully resolves the next-step observable exclusivity structure carried by $R$, in the sense that every admissible continuation in $\Cpsi{R}$ falls under exactly one of $R_1$ or $R_2$, and no admissible continuation remains outside this partition.
\end{enumerate}
\end{definition}

Condition (iii) is essential. It excludes purely formal orthogonal decompositions that do not correspond to a genuine partition of experientially admissible continuations relative to the fixed state \(\Psi\). Thus, admissibility is stronger than orthogonality alone and is state-relative in the sense relevant for the induced weight construction.

Whenever $R = R_1 \oplus R_2$ is an admissible orthogonal refinement relative to \(\Psi\), the projected components satisfy
\[
\phi_R = \phi_{R_1} + \phi_{R_2},
\qquad
\phi_{R_1} \perp \phi_{R_2},
\qquad
\norm{\phi_R}^2 = \norm{\phi_{R_1}}^2 + \norm{\phi_{R_2}}^2.
\]

\subsection{Continuation partition lemma}

The next lemma records the precise bridge between admissible orthogonal refinement in the Hilbert layer and disjoint continuation structure at the level of admissible continuation bundles.

\begin{lemma}[Continuation partition under admissible refinement]
Let $R$ be a robust record sector, and let
\[
R = R_1 \oplus R_2
\]
be an admissible orthogonal refinement relative to \(\Psi\). Then the associated continuation bundles satisfy
\[
\Cpsi{R} = \Cpsi{R_1} \,\dot\cup\, \Cpsi{R_2}.
\]
\end{lemma}

\begin{proof}
Since $R_1,R_2 \subseteq R$, every continuation that realizes $R_1$ or $R_2$ also realizes $R$. Hence
\[
\Cpsi{R_1} \cup \Cpsi{R_2} \subseteq \Cpsi{R}.
\]

Conversely, admissibility means that the refinement $R = R_1 \oplus R_2$ fully resolves the next-step observable exclusivity structure represented by $R$. This follows because condition~(iii) of \cref{def:admissible-refinement} requires the refinement to resolve the full observable exclusivity structure of $R$ at the next step, leaving no continuation in $\Cpsi{R}$ unaccounted for. Therefore every continuation in $\Cpsi{R}$ realizes exactly one of the refined record sectors $R_1$ or $R_2$. Hence
\[
\Cpsi{R} \subseteq \Cpsi{R_1} \cup \Cpsi{R_2}.
\]

Because $R_1$ and $R_2$ represent mutually exclusive record alternatives, no admissible continuation can realize both simultaneously. Thus
\[
\Cpsi{R_1} \cap \Cpsi{R_2} = \varnothing.
\]
Combining the two inclusions with disjointness yields
\[
\Cpsi{R} = \Cpsi{R_1} \,\dot\cup\, \Cpsi{R_2}.
\]
\end{proof}

By finite additivity of the extensive bundle valuation $\mu$ on
disjoint continuation bundles, every admissible orthogonal refinement
automatically induces
\[
\Wpsi{R} = \Wpsi{R_1} + \Wpsi{R_2}.
\]

\begin{definition}[Refinement-stability]
\label{def:refinement-stability}
An induced weight $\Wpsi{\cdot}$ on robust record sectors is called
\emph{refinement-stable} if for every robust record sector $R$ and
every admissible orthogonal refinement
\[
  R = R_1 \oplus R_2
\]
relative to $\Psi$,
\[
  \Wpsi{R} = \Wpsi{R_1} + \Wpsi{R_2}.
\]
\end{definition}

For induced weights arising from an extensive bundle valuation, this
property is automatic under admissible orthogonal refinement by the
continuation partition lemma and finite additivity of $\mu$. The term
\emph{refinement-stable} is retained only to name the property singled
out by the uniqueness theorem. The subsequent sections determine which
non-negative induced weights can have this property on robust record
sectors.

\subsection{Why the additive primitive is placed on continuation bundles}

The additive primitive sits on continuation bundles, not on the
projector lattice. The sector-level additive law used throughout the
rest of the argument is not introduced as a new postulate but is
inherited from admissible continuation partition, as established in
the continuation partition lemma above. The following sections use
exactly this inherited additivity for the profile reduction and the
derivation of the functional equation.

\subsection{Scope of the framework}
The framework fixed above is deliberately narrower than a global measure-theoretic treatment of all projectors on $\HH$. The argument developed in this paper concerns only robust record sectors and their admissible orthogonal refinements. This restriction is essential: the later uniqueness result will be a theorem about internally stable weight assignments on robust record sectors, not a reconstruction of a probability calculus on the full projector lattice.

%% file: sections/03_profile_reduction.tex

\section{Profile Reduction}
\label{sec:profile-reduction}

This section isolates the reduction step that turns sector-level weights into a one-variable function of the projected record component. The key point is that, within the admissible Hilbert record layer, the internally relevant refinement data of a projected component is captured by its admissible binary refinement profile. A norm-based representation of the weight is obtained only once the internal equivalence principle is combined with the norm classification of profiles supplied by admissible binary saturation.

\subsection{Admissible binary refinement profiles}

Let $R$ be a robust record sector and let
\[
\phi_R := \Proj{R}\Psi
\]
be the projected component of the global representational state $\Psi$ in $R$.

We now formalize the refinement data carried by $\phi_R$.

\begin{definition}[Admissible binary refinement profile space]
Let $R$ be a robust record sector and let $\phi_R = \Proj{R}\Psi$. The \emph{admissible binary refinement profile space} of $\phi_R$ is the set
\[
\mathcal{D}(\phi_R)
\]
of all pairs
\[
(r_1,r_2) \in \RR_{\geq 0}^2
\]
for which either
\begin{enumerate}[label=(\roman*)]
\item there exists an admissible orthogonal refinement
\[
R = R_1 \oplus R_2
\]
such that
\[
\norm{\Proj{R_1}\Psi} = r_1,
\qquad
\norm{\Proj{R_2}\Psi} = r_2,
\]
or
\item $(r_1,r_2) = (\norm{\phi_R},0)$, which is included by convention as the trivial self-profile representing the absence of any proper binary refinement.
\end{enumerate}
\end{definition}

Thus, $\mathcal{D}(\phi_R)$ records which binary norm profiles can occur under admissible refinement of the record sector $R$ relative to the state $\Psi$, together with the trivial self-profile, within the admissible refinement structure under consideration.

\subsection{Internal equivalence at the profile level}

The later uniqueness theorem will rely on a structural principle according to which internally indistinguishable record situations must carry the same weight. At the present stage, we only define the corresponding equivalence relation.

\begin{definition}[Binary profile equivalence]
Two pairs $(\Psi,R)$ and $(\Psi',R')$, with $R$ and $R'$ robust record sectors, are called \emph{binary-profile-equivalent} if
\[
\mathcal{D}(\Proj{R}\Psi) = \mathcal{D}(\Proj{R'}\Psi').
\]
\end{definition}

The guiding idea is that binary-profile-equivalent pairs carry the same admissible binary refinement structure and are therefore indistinguishable at the level relevant for internal record stability. The corresponding weight principle will be stated explicitly in the next section.

\subsection{Norm classification of profiles}

The following proposition is an immediate consequence of admissible
binary saturation together with norm equality. Its role is simply to
record that, within the admissible binary refinement class used in this
paper, the admissible binary refinement profile is fully classified by
the norm of the projected component.

\begin{proposition}[Norm determines the admissible binary refinement profile]
\label{prop:norm-determines-profile}
Let $(\Psi,R)$ and $(\Psi',R')$ be two pairs with $R$ and $R'$ robust record sectors inside the admissible Hilbert record layer. Assume that the admissible refinement structure is binary-saturated in the sense that it realizes every admissible binary orthogonal refinement compatible with the squared-norm decomposition of the corresponding projected components. Then
\[
\mathcal{D}(\Proj{R}\Psi) = \mathcal{D}(\Proj{R'}\Psi')
\quad \Longleftrightarrow \quad
\norm{\Proj{R}\Psi} = \norm{\Proj{R'}\Psi'}.
\]
\end{proposition}

\begin{proof}
Set
\[
\phi := \Proj{R}\Psi,
\qquad
\phi' := \Proj{R'}\Psi'.
\]

First assume
\[
\norm{\phi} = \norm{\phi'} = r.
\]
Under admissible binary saturation, both sectors realize exactly the
admissible binary squared-norm decompositions of total size $r^2$.
Hence
\[
\mathcal{D}(\phi) = \mathcal{D}(\phi').
\]

Conversely, suppose
\[
\mathcal{D}(\phi) = \mathcal{D}(\phi').
\]
The degenerate binary profile
\[
(\norm{\phi},0)
\]
belongs to $\mathcal{D}(\phi)$ by the trivial self-profile convention in the definition of $\mathcal{D}(\phi)$. Hence it also belongs to $\mathcal{D}(\phi')$. Likewise,
\[
(\norm{\phi'},0)
\]
belongs to $\mathcal{D}(\phi)$.
Since the total squared norm is fixed across every admissible binary profile of a given projected component (as every pair in $\mathcal{D}(\phi)$ satisfies $r_1^2 + r_2^2 = \norm{\phi}^2$ by the Pythagorean relation for admissible orthogonal refinements), equality of profile spaces implies
\[
\norm{\phi}^2 = \norm{\phi'}^2.
\]
Because norms are non-negative,
\[
\norm{\phi} = \norm{\phi'}.
\]
This proves the claim.
\end{proof}

\begin{remark}
This proposition records only a classification consequence within the
admissible binary refinement class. It does not say that arbitrary
Hilbert decompositions are physically meaningful. It states only that,
once admissible binary saturation is in place, the resulting profile
space is fully classified by the norm of the projected component.
\end{remark}

\subsection{Reduction to a one-variable weight function}

We can now isolate the exact form of the later reduction.

\begin{corollary}[Existence of a norm-reduced weight function]
\label{cor:norm-reduced-weight}
Assume the internal equivalence principle together with admissible binary saturation. Then there exists a function
\[
g : \RR_{\geq 0} \to \RR_{\geq 0}
\]
such that for every global state $\Psi$ and every robust record sector $R$,
\[
\Wpsi{R} = g\!\left(\norm{\Proj{R}\Psi}\right).
\]
\end{corollary}

\begin{proof}
Let
\[
\mathcal{N}
:=
\left\{
\norm{\Proj{R}\Psi}
\;\middle|\;
\Psi \in \HH,\;
R \text{ a robust record sector}
\right\}
\subseteq \RR_{\geq 0}
\]
be the set of realized projected norms. By \cref{prop:norm-determines-profile}, which is available under admissible binary saturation, the binary profile-equivalence class of $(\Psi,R)$ is completely determined by the single non-negative number
\[
\norm{\Proj{R}\Psi}.
\]
If induced weight is constant on binary-profile-equivalence classes, then it depends only on that number. This defines a function
\[
g : \mathcal{N} \to \RR_{\geq 0},
\qquad
g(r) := \Wpsi{R},
\]
for any pair $(\Psi,R)$ satisfying
\[
\norm{\Proj{R}\Psi} = r.
\]
The definition is independent of the chosen representative pair by profile-equivalence invariance.

At this stage representational states are not assumed normalized, and by the positive-scaling closure stated in \cref{sec:framework}, the admissible Hilbert record layer remains closed under positive scalar multiples. Hence, on any non-vacuous admissible record layer, one has \(\mathcal{N}=\RR_{\geq 0}\): if \(R_0\) is any non-zero robust record sector and \(0 \neq \phi \in R_0\), then for every \(r \in \RR_{\geq 0}\) the state
\[
\Psi_r := \frac{r}{\norm{\phi}}\,\phi
\]
satisfies
\[
\norm{\Proj{R_0}\Psi_r} = r.
\]
Thus \(g\) may be regarded as a function
\[
g : \RR_{\geq 0} \to \RR_{\geq 0}.
\]
Since induced weights are non-negative, the codomain of \(g\) is \(\RR_{\geq 0}\).
\end{proof}

\subsection{Role in the main argument}

The present section does not yet determine the form of \(g\). It shows only how the reduction works once two ingredients are combined: the internal equivalence principle, which makes induced weight constant on binary-profile-equivalence classes, and admissible binary saturation, which makes those classes norm-classifiable.

The next section states these two structural conditions explicitly:
\begin{enumerate}[label=(\roman*)]
\item the internal equivalence principle that turns binary profile-equivalence into weight-equivalence,
\item the admissible binary saturation condition that ensures enough binary refinements both to classify profiles by norm and to force the functional equation for \(g\).
\end{enumerate}

The reduction developed in this section acquires genuine force only if admissible refinement structure is not trivial. If every record situation carried only its trivial self-profile, then binary-profile equivalence would cease to induce a non-trivial internal classification, and invariance at that level would impose no substantive restriction on induced weight. What is therefore still needed is twofold: first, an invariance principle that limits weight-relevance to internally accessible refinement structure rather than representational surplus; second, a richness condition ensuring that admissible refinement structure is strong enough to support a constraining functional equation. The next section states these two requirements explicitly in the form of the internal equivalence principle and admissible binary saturation. In this way, the present section isolates the correct internal structural carrier, while the next section supplies the additional assumptions under which that carrier becomes theorematically decisive.

Once these are in place, the continuation partition lemma from the
framework section yields the additive refinement law, and the
quadratic assignment follows.

It is worth making the logical order explicit to forestall a potential
circularity concern. \Cref{cond:internal-equivalence} is stated at
the level of binary refinement profiles, before any connection to
norms has been established. That binary-profile-equivalent pairs are
norm-classified is the content of \cref{prop:norm-determines-profile},
which is a later consequence of admissible binary saturation. The
quadratic assignment then follows from the Cauchy lemma, applied to
the functional equation that admissible binary saturation makes
available. \Cref{cond:internal-equivalence} does not assume
norm-dependence of weight; it is one input to the argument that
eventually forces it.

%% file: sections/04_structural_conditions.tex
\section{Two Structural Conditions}
\label{sec:structural-conditions}

The framework and the profile reduction established so far do not yet determine the form of the induced weight. They identify the correct objects and show how the sector-level assignment reduces to a one-variable function of the norm of the projected component only once internal equivalence is combined with admissible binary saturation.

The remaining argument rests on two explicit structural conditions. They should not be read as hidden background assumptions. On the contrary, they mark the precise points at which the present theorem restricts the class of admissible record sectors under consideration.

\subsection{Internal equivalence principle}

The profile spaces introduced in \cref{sec:profile-reduction} provide
the internal structural carrier on which record situations can be
compared. For this comparison to yield a non-arbitrary induced weight,
the weight must be constrained to depend only on that internally
accessible structure, not on outer redescription. The first condition
makes this requirement precise.

The first condition states that internally indistinguishable record situations must carry the same induced weight.

\begin{condition}[Internal equivalence principle]
\label{cond:internal-equivalence}
We say that an induced weight satisfies the \emph{internal equivalence
principle} if binary-profile-equivalent pairs carry the same induced
weight:
\[
\mathcal{D}(\Proj{R}\Psi) = \mathcal{D}(\Proj{R'}\Psi')
\quad \Longrightarrow \quad
\Wpsi{R} = W_{\Psi'}\!\left(R'\right).
\]
\end{condition}

\begin{remark}[Internal equivalence as representation invariance]
The internal equivalence principle is best read as an invariance
requirement on induced weight. It says that induced weight must depend
only on internally accessible admissible refinement structure, not on
external coding, ambient embedding, or other representational
surplus. Without such an internal-accessibility requirement, induced
weight in the present framework ceases to be fixed by admissible
refinement structure alone and becomes sensitive to representational
surplus.

In the present framework, binary-profile equivalence is the relevant
notion of internal sameness. The point is not merely that matching
profiles are necessary for internal sameness, but that admissible
binary refinement structure is the full internal structural carrier
that the present theorem allows induced weight to see. Once external
coding, ambient embedding, and other representational surplus are set
aside, no further internally accessible weight-relevant datum remains
within the theorem target beyond the admissible binary refinement
profile itself. If two record situations support the same admissible
binary refinement profiles, then no difference between them is
available at the structural level that the weight is meant to track.
A weight assignment that nevertheless distinguished them would fail to
be intrinsic to the record structure itself. Rejecting the principle
therefore comes at a definite price: one must allow induced weight to
depend on some further feature not fixed by the admissible internal
refinement profile, and hence on representational surplus rather than
on record structure alone. In this sense, binary-profile equivalence
is the natural internal equivalence for the present problem: it tracks
exactly the distinctions carried by the internally accessible
admissible refinement structure. Any finer relation would reintroduce
representational surplus, whereas any coarser relation would collapse
distinctions already present in the internal profile.

The principle is therefore selective rather than measure-theoretic. It
does not postulate a global measure on projectors. It constrains which
features of a record situation are allowed to matter for induced
weight.
\end{remark}

\begin{remark}[What the internal equivalence principle does and does not assume]
The internal equivalence principle does not assume that induced weight is a function of norm. The norm classification of profiles is a later consequence of admissible binary saturation, established in \cref{prop:norm-determines-profile}. What the principle does assume is profile-sufficiency: induced weight is fixed by the admissible binary refinement structure of the record situation and cannot depend on representational features not encoded in that profile.
\end{remark}

\subsection{Admissible binary saturation}

The second condition ensures that the admissible record sector is sufficiently saturated to realize the full family of norm-compatible binary refinements needed for the functional equation.

\begin{condition}[Admissible binary saturation]
\label{cond:binary-saturation}
We say that the admissible record layer satisfies \emph{admissible binary saturation} if for every global state $\Psi$, every robust record sector $R$, and every pair of non-negative numbers $r_1,r_2 \in \RR_{\geq 0}$ satisfying
\[
r_1^2 + r_2^2 = \norm{\Proj{R}\Psi}^2,
\]
there exists an admissible orthogonal refinement
\[
R = R_1 \oplus R_2
\]
such that
\[
\norm{\Proj{R_1}\Psi} = r_1,
\qquad
\norm{\Proj{R_2}\Psi} = r_2.
\]
\end{condition}

This condition is stronger than mere Hilbert-space decomposability. It is a statement about the admissible refinement class, not about arbitrary orthogonal splittings of a subspace. Its role is to ensure that the additive refinement law derived from continuation partition can be instantiated for every norm-compatible binary profile. The term ``binary saturation'' is meant to emphasize that the admissible refinement structure is saturated enough to realize every such binary profile within the stated scope.

\begin{remark}
Admissible binary saturation is not merely a scope restriction. It
identifies the structural threshold at which the quadratic assignment
becomes the only viable non-negative refinement-stable induced weight.
The theorem's contribution is precisely to isolate this threshold
explicitly.

The condition remains a genuine restriction on the domain of
application. Whether physically relevant systems in dimension greater
than two satisfy admissible binary saturation is an open structural
question that this paper does not resolve. Decoherence-stabilized
pointer sectors are natural candidates for the kind of record
structure the theorem targets, but whether such sectors realize every
norm-compatible binary refinement as an admissible orthogonal
decomposition depends on physical details that lie outside the present
scope.

The theorem should therefore be read as a threshold result: it
identifies the structural requirement that forces the quadratic
assignment and leaves open, as a separate empirical and structural
question, which systems meet that requirement in dimension greater
than two. The two-outcome spin example of
\cref{subsec:spin-example} illustrates the minimal binary setting
targeted by the theorem; it does not show how large that domain is.
\end{remark}

\begin{remark}[Why refinement richness matters]
\label{rem:equal-split-counterexample}
Binary saturation is not merely a technical convenience. To see this, consider a toy refinement class in which, for each total norm $s \geq 0$, the only admissible non-trivial binary refinement is the equal split
\[
  \left(\frac{s}{\sqrt{2}}, \frac{s}{\sqrt{2}}\right).
\]
Then the continuation-partition law yields only the relation
\[
  g(s) = 2\,g\!\left(\frac{s}{\sqrt{2}}\right)
  \qquad \text{for all } s \in \RR_{\geq 0},
\]
rather than the full quadratic functional equation.

This restricted relation does not force the quadratic form. For any fixed
$0 < \varepsilon < 1$, define
\[
  g_{\varepsilon}(0) := 0,
\]
and for $s > 0$,
\[
  g_{\varepsilon}(s)
  := s^2\!\left(1 + \varepsilon \sin\!\bigl(4\pi \log_2 s\bigr)\right).
\]
Then $g_{\varepsilon}$ is continuous and non-negative on
$\RR_{\geq 0}$, and it satisfies
\[
  g_{\varepsilon}(s)
  = 2\,g_{\varepsilon}\!\left(\frac{s}{\sqrt{2}}\right)
  \qquad \text{for all } s \in \RR_{\geq 0},
\]
since
\[
  \sin\!\bigl(4\pi(\log_2 s - \tfrac{1}{2})\bigr)
  = \sin\!\bigl(4\pi \log_2 s - 2\pi\bigr)
  = \sin\!\bigl(4\pi \log_2 s\bigr).
\]
But $g_{\varepsilon}$ is not of the form $c\,s^2$. Thus continuity alone
does not restore uniqueness when the admissible refinement class is too
sparse. What matters is sufficient refinement richness to force the full
functional equation.
\end{remark}

\subsection{Reduction to the functional equation}

We now collect the consequences of the previous section together with the two structural conditions stated above.

By \cref{cor:norm-reduced-weight}, once the internal equivalence principle is combined with admissible binary saturation, there exists a function
\[
g : \RR_{\geq 0} \to \RR_{\geq 0}
\]
such that
\[
\Wpsi{R} = g\!\left(\norm{\Proj{R}\Psi}\right)
\]
for every global state $\Psi$ and every robust record sector $R$.

Let
\[
R = R_1 \oplus R_2
\]
be an admissible orthogonal refinement. By the continuation partition lemma from the framework section,
\[
\Wpsi{R} = \Wpsi{R_1} + \Wpsi{R_2}.
\]
Using the norm-reduced form, this becomes
\[
g\!\left(\norm{\Proj{R}\Psi}\right)
=
g\!\left(\norm{\Proj{R_1}\Psi}\right)
+
g\!\left(\norm{\Proj{R_2}\Psi}\right).
\]
Since admissible orthogonal refinement satisfies the Pythagorean relation,
\[
\norm{\Proj{R}\Psi}^2
=
\norm{\Proj{R_1}\Psi}^2 + \norm{\Proj{R_2}\Psi}^2,
\]
we obtain
\[
g\!\left(\sqrt{
\norm{\Proj{R_1}\Psi}^2 + \norm{\Proj{R_2}\Psi}^2
}\right)
=
g\!\left(\norm{\Proj{R_1}\Psi}\right)
+
g\!\left(\norm{\Proj{R_2}\Psi}\right).
\]

Admissible binary saturation now allows this relation to be realized for every pair $r_1,r_2 \in \RR_{\geq 0}$ with admissible total norm. Since representational states are not assumed normalized at this stage, and since the admissible Hilbert record layer is closed under positive scalar multiples as stated in \cref{sec:framework}, the realized norm domain is \(\RR_{\geq 0}\). Therefore the following functional equation holds on all of $\RR_{\geq 0}$:
\begin{equation}
\label{eq:quadratic-functional-equation}
g\!\left(\sqrt{r_1^2+r_2^2}\right) = g(r_1) + g(r_2)
\qquad
\text{for all } r_1,r_2 \in \RR_{\geq 0}.
\end{equation}
Since induced weights are non-negative by definition, the function $g$ satisfies $g(r) \geq 0$ for all $r \in \RR_{\geq 0}$. This regularity condition will be used in the next section to exclude pathological solutions of the functional equation.

\subsection{Interpretive status of the conditions}

The argument has now been reduced to a single functional equation, but this reduction is only as strong as the two structural conditions that support it.

The internal equivalence principle is the deeper of the two. It expresses the claim that induced weight must be fixed by internally accessible refinement structure and not by representational surplus. Admissible binary saturation is more technical, but it is not merely cosmetic: it defines the class of record sectors on which the uniqueness result will hold.

Neither condition should be hidden in the prose of the proof. Their role is precisely to make explicit where the present route differs from standard derivations. The main theorem will show that, under these conditions, the quadratic assignment is forced. It will not claim that the conditions themselves hold in every conceivable representation layer.

\subsection{Transition to the main theorem}

At this point the framework, the profile reduction, and the two
structural conditions have reduced the problem to a single functional
equation. The next section uses this reduction to establish the
uniqueness of non-negative refinement-stable induced weights on robust
record sectors. The explicit quadratic form is the direct expression
of this uniqueness.

%% file: sections/05_main_theorem.tex

\section{Main Theorem}
\label{sec:main-theorem}

We now combine the results of the preceding sections to establish the
uniqueness of non-negative refinement-stable induced weights on robust
record sectors. The explicit quadratic form obtained below is the
direct expression of this uniqueness.

\begin{lemma}[Non-negative additive functions on $\RR_{\geq 0}$ are linear]
\label{lem:cauchy}
Let $f : \RR_{\geq 0} \to \RR_{\geq 0}$ satisfy
\[
  f(u + v) = f(u) + f(v)
  \qquad \text{for all } u, v \in \RR_{\geq 0}.
\]
Then there exists a constant $c \geq 0$ such that $f(x) = cx$ for all
$x \in \RR_{\geq 0}$.
\end{lemma}

\begin{proof}
Set $c := f(1) \geq 0$. For each positive integer $n$, additivity gives
$f(n) = nf(1) = cn$. Since $f(1) = f(n \cdot \tfrac{1}{n}) = n
f(\tfrac{1}{n})$, we obtain $f(\tfrac{1}{n}) = \tfrac{c}{n}$. Therefore
$f(q) = cq$ for every non-negative rational $q$.

Since $f$ is non-negative and additive, it is monotone: for $0 \leq x
\leq y$ we have $f(y) = f(x) + f(y - x) \geq f(x)$. For arbitrary $x
\in \RR_{\geq 0}$, choose rational sequences $q_k^- \uparrow x$ and
$q_k^+ \downarrow x$. Monotonicity gives
\[
  c q_k^- = f(q_k^-) \leq f(x) \leq f(q_k^+) = c q_k^+.
\]
Passing to the limit yields $f(x) = cx$.
\end{proof}

\begin{theorem}[Structural uniqueness of refinement-stable induced weight]
\label{thm:quadratic-uniqueness}
Assume:
\begin{enumerate}[label=(\roman*)]
\item an extensive bundle valuation $\mu$ on the family of admissible
  continuation bundles, satisfying finite additivity on disjoint
  admissible bundles, as specified in \cref{sec:framework};
\item the internal equivalence principle, \cref{cond:internal-equivalence};
\item admissible binary saturation, \cref{cond:binary-saturation}.
\end{enumerate}
Then the quadratic assignment is the only non-negative
refinement-stable induced weight on robust record sectors.
Equivalently, there exists a constant $c \geq 0$ such that for every
global state $\Psi$ and every robust record sector $R$,
\[
  \Wpsi{R} = c\,\norm{\Proj{R}\Psi}^2.
\]
\end{theorem}

\begin{proof}
By \cref{cor:norm-reduced-weight}, there exists a function
$g : \RR_{\geq 0} \to \RR_{\geq 0}$ such that
\[
  \Wpsi{R} = g\!\left(\norm{\Proj{R}\Psi}\right)
\]
for every pair $(\Psi, R)$. By \cref{eq:quadratic-functional-equation},
$g$ satisfies
\[
  g\!\left(\sqrt{r_1^2 + r_2^2}\right) = g(r_1) + g(r_2)
  \qquad \text{for all } r_1, r_2 \in \RR_{\geq 0}.
\]
Define $f : \RR_{\geq 0} \to \RR_{\geq 0}$ by $f(x) := g(\sqrt{x})$.
Then for all $u, v \in \RR_{\geq 0}$,
\[
  f(u + v)
  = g\!\left(\sqrt{u + v}\right)
  = g\!\left(\sqrt{(\sqrt{u})^2 + (\sqrt{v})^2}\right)
  = g(\sqrt{u}) + g(\sqrt{v})
  = f(u) + f(v),
\]
so $f$ is additive on $\RR_{\geq 0}$. Since induced weights arise from
a non-negative extensive bundle valuation, $f(x) \geq 0$ for all
$x \geq 0$. By \cref{lem:cauchy}, there exists $c \geq 0$ such that
$f(x) = cx$.
Returning to $g$,
\[
  g(r) = f(r^2) = cr^2 \qquad \text{for all } r \in \RR_{\geq 0},
\]
and therefore
\[
  \Wpsi{R} = c\,\norm{\Proj{R}\Psi}^2.
\]
This is the quadratic form asserted. Since every non-negative
refinement-stable induced weight must satisfy this equation by the
preceding reduction, the quadratic assignment is the only such weight.
\end{proof}

In operational settings, one may regard exact realization of every norm-compatible binary split as stronger than needed: it is enough that admissible refinements approximate such splits arbitrarily well, provided the induced profile function varies continuously with the realized norm.

\begin{proposition}[Dense saturation under continuity]
\label{prop:dense-saturation}
Let $g:\RR_{\geq 0}\to\RR_{\geq 0}$ be continuous. Assume that for every $s\ge 0$ there exists a dense subset $A_s\subseteq[0,s]$ such that
\[
g(s)=g(t)+g\!\left(\sqrt{s^{2}-t^{2}}\right)
\qquad\text{for all } t\in A_s .
\]
Then there exists $c\ge 0$ such that
\[
g(r)=c\,r^{2}
\qquad\text{for all } r\ge 0.
\]
\end{proposition}

\begin{proof}
Define $f:\RR_{\geq 0}\to\RR_{\geq 0}$ by
\[
f(x):=g(\sqrt{x}).
\]
Since $g$ is continuous, so is $f$. Fix $s\ge 0$ and set
\[
B_s:=\{t^{2}: t\in A_s\}\subseteq [0,s^{2}].
\]
Because $A_s$ is dense in $[0,s]$, the set $B_s$ is dense in $[0,s^{2}]$. For every $x\in B_s$, writing $x=t^{2}$ with $t\in A_s$, the assumed identity gives
\[
f(s^{2})=f(x)+f(s^{2}-x).
\]
The map
\[
x \mapsto f(x)+f(s^{2}-x)
\]
is continuous on $[0,s^{2}]$. Since
\[
f(s^{2})=f(x)+f(s^{2}-x)
\]
holds on the dense set $B_s$, it extends to all $x\in[0,s^{2}]$.
Hence, for arbitrary $u,v\ge 0$, taking $s=\sqrt{u+v}$ and $x=u$, we obtain
\[
f(u+v)=f(u)+f(v).
\]
Thus $f$ is additive on $\RR_{\geq 0}$. Since $f$ is non-negative, it is monotone on $\RR_{\geq 0}$, and by \cref{lem:cauchy} there exists $c\ge 0$ such that
\[
f(x)=c\,x
\qquad\text{for all } x\ge 0.
\]
Returning to $g$, we obtain
\[
g(r)=f(r^{2})=c\,r^{2}
\qquad\text{for all } r\ge 0.
\]
\end{proof}

This shows that full admissible binary saturation can be weakened to dense admissible saturation once continuity of the weight function $g$ is assumed.

\begin{remark}[Why continuity can be plausible]
Continuity of the profile function $g$ is not derived by the present
framework. In the dense-saturation route, it is an additional
regularity assumption. Its appeal is operational rather than
measure-theoretic: if admissible refinements approximate a
norm-compatible binary split arbitrarily well, then it is natural to
require that the induced weight should not undergo finite jumps under
arbitrarily small changes in the realized norm.

This makes continuity plausible when induced weight is meant to track
stable record structure rather than threshold-sensitive coding
artifacts. But the assumption should not be treated as free. If one
allows genuinely discontinuous dependence on norm, dense admissible
saturation need not determine the weight uniquely. Moreover, as
\cref{rem:equal-split-counterexample} shows, continuity alone still
does not suffice when the admissible refinement class is too sparse.
\end{remark}

\begin{corollary}[Born assignment as normalized realization]
\label{cor:normalized-quadratic-weight}
Assume in addition:
\begin{enumerate}[label=(\alph*)]
\item the induced weight is normalized: the total admissible next-step
  continuation bundle carries total induced weight $1$;
\item these sectors form a complete orthogonal decomposition of the
  relevant next-step content, so that
  \[
    \sum_i \norm{\Proj{R_i}\Psi}^2 = \norm{\Psi}^2 = 1.
  \]
\end{enumerate}
Then
\[
  \Wpsi{R} = \norm{\Proj{R}\Psi}^2
\]
for every global state $\Psi$ and every robust record sector $R$ in
the decomposition.
\end{corollary}

\begin{proof}
By \cref{thm:quadratic-uniqueness}, the induced weight has the form
$\Wpsi{R} = c\,\norm{\Proj{R}\Psi}^2$ with a constant $c \geq 0$.
By assumption~(b), the squared norms of the projected components sum
to $\norm{\Psi}^2 = 1$. The total induced weight of the admissible
next-step continuation bundle therefore equals
\[
  \sum_i \Wpsi{R_i} = c \sum_i \norm{\Proj{R_i}\Psi}^2 = c.
\]
By assumption~(a), this total equals $1$, so $c = 1$. Therefore
\[
  \Wpsi{R} = \norm{\Proj{R}\Psi}^2.
\]
\end{proof}

\begin{remark}
The theorem does not construct a probability measure on the full projector lattice of \(\HH\). It establishes a uniqueness result only for induced weights on robust record sectors inside an admissible Hilbert record layer.
\end{remark}

\begin{remark}
The proof uses no appeal to Gleason's theorem, no decision-theoretic axiom, and no envariance argument. The quadratic assignment is forced directly by sector-level continuation partition, norm-profile reduction, and the two structural conditions isolated in \cref{sec:structural-conditions}.
\end{remark}

\begin{remark}[Compression viewpoint]
The theorem can also be read as a compression statement. Once admissible induced weights are reduced to a norm-profile function $g$, the entire admissible weight structure collapses to the one-parameter form
\[
  g(r) = c\,r^2.
\]
No further dependence on the internal realization of a robust record sector survives beyond the norm of its projected component. Under the normalization condition of \cref{cor:normalized-quadratic-weight}, even the constant $c$ is fixed, so the induced weight is determined uniquely by squared norm alone.
\end{remark}

The next section clarifies the exact scope of this result, the sense in which it is conditional, and the limits of what has and has not been shown.

%% file: sections/06_scope_and_limits.tex
\section{Scope and Limits}
\label{sec:scope-limits}

The result established in \cref{thm:quadratic-uniqueness} is deliberately conditional and tightly focused. Its point is not to solve every surrounding foundational problem at once, but to isolate the precise structural threshold at which the quadratic assignment becomes unavoidable. In that sense, the narrowness of the theorem is part of its content: it identifies exactly which assumptions do the forcing work and exactly where the uniqueness claim begins.

\subsection{What the theorem establishes}

The main theorem shows that, within an admissible Hilbert record layer, and under the two structural conditions stated in \cref{cond:internal-equivalence,cond:binary-saturation}, the induced weight on robust record sectors is forced to take the quadratic form
\[
\Wpsi{R} = c\,\norm{\Proj{R}\Psi}^2.
\]
With the additional normalization stated in \cref{cor:normalized-quadratic-weight}, this becomes
\[
\Wpsi{R} = \norm{\Proj{R}\Psi}^2.
\]

The content of the theorem is therefore not that one may \emph{choose} the quadratic assignment, but that, once the stated structural threshold is met, no other non-negative refinement-stable assignment remains compatible with the framework on the stated class of record sectors.

\subsection{What the theorem does not establish}

It is equally important to state what has \emph{not} been shown.

First, the theorem does not derive a probability measure on the full projector lattice of \(\HH\). The argument is restricted from the outset to robust record sectors and admissible orthogonal refinements among them. This restriction is essential to the logic of the paper.

Second, the theorem does not derive Hilbert structure itself. The admissible representation layer is assumed to have Hilbert-sector form. The present argument then shows what weight assignment is forced \emph{within} that layer. The deeper question of why admissible representation layers should be Hilbertian is a separate foundational problem. This is not an evasion but a deliberate scope decision: the present argument is designed to show what is forced \emph{within} a Hilbert record layer, independently of whichever deeper justification one might accept for that layer. The conditional structure makes the result applicable across multiple foundational positions on the Hilbert-structure question.

Third, the theorem does not claim that every orthogonal splitting of a Hilbert subspace is physically meaningful. Only admissible orthogonal refinements enter the argument. Admissibility is stronger than orthogonality: it requires faithful representation of the observable next-step exclusivity structure of the record sector being refined.

Fourth, the theorem does not eliminate the need for substantive structural assumptions. On the contrary, it makes them explicit. The proof depends essentially on the internal equivalence principle and on admissible binary saturation. These are not cosmetic additions. They define the domain on which the uniqueness claim is valid.

Fifth, the theorem does not establish that physically relevant systems in dimension greater than two fall within its domain of application. The two-outcome spin example of \cref{rem:spin-example} illustrates only the minimal two-dimensional setting targeted by the framework. Whether robust record sectors satisfying admissible binary saturation, or the weaker dense saturation condition of \cref{prop:dense-saturation} together with continuity, exist for physically relevant systems in higher dimensions remains an open question. The theorem identifies the structural threshold that forces the quadratic assignment; it does not determine which physical systems realize that threshold. This is the precise sense in which the result is conditional: not merely that it assumes its premises, but that the physical scope of those premises remains to be mapped.

The conditional structure of the argument can be summarized by the following failure map.

\begin{center}
\renewcommand{\arraystretch}{1.15}
\begin{tabular}{%
  >{\raggedright\arraybackslash}p{0.24\linewidth}
  >{\raggedright\arraybackslash}p{0.28\linewidth}
  >{\raggedright\arraybackslash}p{0.34\linewidth}}
\textbf{If this is removed} & \textbf{What remains} & \textbf{What is lost} \\[0.3em]
Extensive bundle valuation with finite additivity on disjoint continuation bundles
& admissible sectors and refinements may still be defined
& no sector-level additive law, hence no route to the functional equation \\[0.5em]
Internal equivalence
& refinement additivity may still hold on admissible splits
& no reduction to a one-variable profile function $g$, hence no scalar uniqueness problem \\[0.5em]
Sufficient admissible refinement richness
& only partial additivity relations on the realized refinement class
& no uniqueness of the induced weight in general; see \cref{rem:equal-split-counterexample} \\[0.5em]
Normalization
& the quadratic form $\Wpsi{R} = c\,\norm{\Proj{R}\Psi}^2$
& the constant $c$ remains undetermined, so the Born assignment is not yet fixed
\end{tabular}
\end{center}

The framework also carries four distinct structural loads. First,
exclusivity of record alternatives ensures that continuation bundles
form a canonical disjoint carrier for $\mu$; without it, finite
additivity loses its unambiguous domain of application. Second,
continuation structure with short-horizon stability makes
$W_\Psi(R)=\mu(C_\Psi(R))$ genuinely induced rather than directly
postulated; without it, sector weight becomes primitive and the
problem type changes. Third, internal refinement richness supplies the
non-trivial profile structure needed for the later functional
equation; without sufficient richness, profile equivalence becomes too
sparse to constrain induced weight. Fourth, internal accessibility of
admissible refinement structure is what prevents induced weight from
depending on representational surplus; without it, the internal
equivalence principle loses its ground.

\subsection{Why the conditional form matters}

The value of the conditional form is that it makes the threshold
structure of the result explicit while keeping mathematical
uniqueness separate from physical applicability. The theorem settles
the former on its stated domain and leaves the latter to the
structural and empirical question of which systems instantiate the
required record structure.

\subsection{The role of the internal equivalence principle}

The internal equivalence principle, \cref{cond:internal-equivalence},
is the deeper of the two structural conditions. It requires induced
weight to be fixed by internally accessible admissible refinement
structure rather than by representational surplus. Rejecting it means
allowing internally indistinguishable record situations to carry
different induced weights, so that weight is no longer fixed by
admissible refinement structure alone.

\subsection{The role of admissible binary saturation}

Admissible binary saturation, \cref{cond:binary-saturation}, is the
richness condition that turns continuation-partition additivity into a
constraining functional equation for \(g\). It should not be read as a
demand that arbitrary orthogonal splittings be physically realizable,
but as the requirement that a robust record sector carry enough stable
internal differentiation to realize norm-compatible binary
refinements, exactly or at least densely. Without such refinement
richness, the induced weight need not be uniquely determined. As
\cref{prop:dense-saturation} shows, full admissible binary saturation
may be weakened to dense admissible saturation if continuity of the
profile function is added.

\subsection{No ontological ranking of branches}

The theorem should also not be read as assigning different \emph{degrees of reality} to different record sectors. What has been established is a uniqueness result for induced weights under refinement stability conditions. That is a structural statement about admissible weight assignments, not an ontological ranking theorem.

For that reason, the language of this paper remains deliberately neutral. It speaks of induced weight on robust record sectors, not of branch reality, branch substance, or metaphysical thickness. In that sense, the theorem is interpretation-neutral: it does not privilege Everettian branching over collapse-based, relational, or pragmatist readings, provided they admit the relevant record structure on the stated domain.

\subsection{Why the restriction to record sectors is substantive}

One might worry that restricting the theorem to robust record sectors makes the result too narrow. We take the opposite view. The restriction is exactly what gives the theorem a distinct foundational target.

The present question is not how to assign a measure to arbitrary formal subspaces, but which weight assignments remain viable on structures that can function as stable, internally accessible records. That is a sharper and more physically motivated question. The resulting theorem is narrower than a global measure theorem, but also closer to the problem of stable internal record structure.

\subsection{A schematic decoherence-stabilized record model}
\label{subsec:decoherence-schematic}

The following sketch illustrates how the framework of this paper can
be read in a setting closer to physical practice. It does not derive
decoherence, does not establish a pointer basis, and does not prove
that admissible binary saturation holds for any specific physical
system. The theorem itself does not depend on any specific
decoherence model; the present sketch only exhibits one physically
familiar way of reading the formal notions introduced above. Its
purpose is only to make the abstract framework \emph{physically
legible} in a standard schematic setting.

\emph{Setup.}
Consider a finite-dimensional system $S$ coupled to an environment
$E$, and suppose that the interaction selects a decoherence-stabilized
record structure with macroscopically distinct alternatives indexed by
$i$ \cite{Zurek2003Decoherence,Schlosshauer2007Decoherence}.
Rather than identifying record sectors with one-dimensional
pointer states, let $R_i$ denote coarse-grained subspaces associated
with stable record alternatives of the composite $SE$-system, or of an
effective record-bearing subsystem. For present purposes, these
sectors are treated as an idealized exactly orthogonal and
short-horizon stable decomposition, even though in realistic
decoherence models such structure is typically only approximate.
In this setting,
\cref{def:robust-record-sector} receives the following reading:
\emph{internal discriminability} corresponds to the record alternatives
being operationally distinguishable, \emph{short-horizon persistence}
corresponds to their stability over the timescale relevant to the next
step of record use, and \emph{admissible refinement closure}
corresponds to the possibility of resolving a coarse record
alternative into finer stable sub-records.

\emph{Admissible continuations and refinements.}
The continuation bundles $\Cpsi{R_i}$ can then be read as the class of
post-interaction continuations of the composite dynamics in which the
record content remains consistent with the sector $R_i$. An admissible
orthogonal refinement
\[
R_i = R_{i,1} \oplus R_{i,2}
\]
corresponds schematically to a further stable discrimination within
the same coarse record alternative, provided the refined sectors are
again mutually exclusive and robust enough to function as record
bearers.

\emph{What is structurally set, what is plausible, what remains open.}
What is \emph{set} by this schematic is only the reading of robust
record sectors as decoherence-stabilized record subspaces, of
continuation bundles as record-consistent continuations, and of
admissible refinements as further stable discriminations within a
coarse record structure.
What is \emph{plausible} is that such a
setting may plausibly support a sufficiently rich family of admissible
refinements,
and that dense admissible saturation may therefore be
more realistic than full saturation when exact realization of every
norm-compatible split is operationally too strong. What remains
\emph{open} is whether admissible binary saturation, or even the
weaker dense saturation condition of \cref{prop:dense-saturation},
holds for any concrete physical system. The theorem identifies the
structural threshold, the schematic makes it physically legible, and
the gap between legibility and certification remains a separate
empirical and structural question.

\begin{remark}[A concrete candidate regime for dense saturation]
A plausible physical regime for the weaker dense-saturation condition is a high-dimensional decoherence-stabilized record sector with redundant environmental encoding \cite{OllivierPoulinZurek2004,Zurek2009QuantumDarwinism}. One should think here not of a one-dimensional pointer state, but of a coarse macroscopic record, such as a detector outcome or pointer reading, that is stably carried by many microscopically distinct apparatus-environment configurations.

In such a setting, an admissible refinement need not arise from arbitrary one-dimensional splittings. It can instead arise from further stable discrimination within the same coarse record sector by grouping families of robust micro-record subspaces that preserve the same coarse record content. When the number of such approximately independent record-bearing components is large, and when the sector admits sufficiently fine robust sub-record decompositions whose binary groupings remain admissible, \cref{prop:local-dense-saturation-criterion} provides a concrete route by which the realized admissible binary refinement norms can become dense in the interval allowed by the total sector norm. As in the schematic model of \cref{subsec:decoherence-schematic}, exact orthogonality of the sub-records is an idealization; in practice the relevant structure is only approximately orthogonal, and the dense-saturation condition is stated here as applying to the idealized version.

This does not establish dense admissible saturation for any concrete model, and it does not show that every norm-compatible split is physically realizable. Its role is narrower: it identifies a non-trivial physical regime in which the dense version of the saturation condition is at least credible. The natural candidate regime is that of redundant macroscopic records supported by many approximately orthogonal environmental subrecords in a large apparatus-environment sector.
\end{remark}

\subsection{A local structural criterion for dense admissible saturation}

The candidate regime just described can be sharpened into a local
structural criterion. The point is not to certify a concrete physical
model, but to isolate an explicit microstructural hypothesis under
which dense admissible saturation follows for a given pair
\((\Psi,R)\).

\begin{proposition}[Local criterion for dense admissible saturation]
\label{prop:local-dense-saturation-criterion}
Fix a global state \(\Psi\) and a robust record sector \(R\), and write
\[
s := \norm{\Proj{R}\Psi}.
\]
Assume that there exists a sequence of finite orthogonal decompositions
\[
R = \bigoplus_{i=1}^{N_n} S_i^{(n)}
\qquad (n\in\mathbb{N})
\]
such that each \(S_i^{(n)}\) is itself a robust record sector, and that
the following hold for every \(n\):
\begin{enumerate}[label=(\roman*)]
\item for every subset \(A \subseteq \{1,\dots,N_n\}\), the grouped sectors
\[
R_A^{(n)} := \bigoplus_{i\in A} S_i^{(n)},
\qquad
R_{A^c}^{(n)} := \bigoplus_{i\notin A} S_i^{(n)}
\]
form an admissible orthogonal refinement of \(R\);
\item with
\[
w_i^{(n)} := \norm{\Proj{S_i^{(n)}}\Psi}^2,
\]
one has
\[
\max_{1\le i\le N_n} w_i^{(n)}
\le
\varepsilon_n \norm{\Proj{R}\Psi}^2
\]
for some sequence \(\varepsilon_n \to 0\).
\end{enumerate}
Then the set of realized norm values
\[
D_{\Psi,R}
:=
\left\{
\norm{\Proj{R_A^{(n)}}\Psi}
\;\middle|\;
n\in\mathbb{N},\;
A\subseteq\{1,\dots,N_n\}
\right\}
\]
is dense in \([0,s]\). In particular, if the same criterion holds for
every pair \((\Psi,R)\) in the relevant admissible domain, then the
dense-realizability hypothesis used in \cref{prop:dense-saturation} is
satisfied on that domain.
\end{proposition}

\begin{proof}
If \(s=0\), then \(D_{\Psi,R}=\{0\}\), so there is nothing to prove.
Assume \(s>0\), and set
\[
W := s^2 = \norm{\Proj{R}\Psi}^2.
\]
For each \(n\) and each subset \(A \subseteq \{1,\dots,N_n\}\),
orthogonality gives
\[
\norm{\Proj{R_A^{(n)}}\Psi}^2
=
\sum_{i\in A} w_i^{(n)}.
\]
Fix \(x\in[0,W]\). If \(x=W\), choose \(A=\{1,\dots,N_n\}\), so there
is nothing to prove. Assume \(x<W\), and for each \(n\) choose among
all subsets \(A\) satisfying
\[
\sum_{i\in A} w_i^{(n)} \le x
\]
one for which the sum is maximal. Write
\[
x_A^{(n)} := \sum_{i\in A} w_i^{(n)}.
\]
We claim that
\[
0 \le x - x_A^{(n)} < \max_i w_i^{(n)}.
\]
Indeed, if \(x - x_A^{(n)} \ge \max_i w_i^{(n)}\), then \(A\) cannot be
the full index set, since \(x_A^{(n)}<W\). Hence there exists
\(j\notin A\), and for every such \(j\) one has
\[
w_j^{(n)} \le \max_i w_i^{(n)} \le x - x_A^{(n)},
\]
so
\[
x_A^{(n)} + w_j^{(n)} \le x,
\]
contradicting the maximality of \(x_A^{(n)}\). Therefore
\[
0 \le x - x_A^{(n)} < \max_i w_i^{(n)}
\le
\varepsilon_n W.
\]
Since \(\varepsilon_n \to 0\), the realized quadratic sums are dense in
\([0,W]\). Because the square-root map is a homeomorphism from
\([0,W]\) to \([0,s]\), the realized norm values in \(D_{\Psi,R}\) are
dense in \([0,s]\).
\end{proof}

\begin{remark}[A worked finite illustration of the local criterion]
The mechanism behind \cref{prop:local-dense-saturation-criterion} can
be seen in a simple finite example. Let \(R\) be a robust record
sector relative to \(\Psi\) with
\[
\norm{\Proj{R}\Psi}^2 = 1,
\]
and suppose that
\[
R = S_1 \oplus S_2 \oplus S_3 \oplus S_4
\]
is a finite orthogonal decomposition in which each \(S_i\) is itself a
robust record sector and every binary grouping of these four sectors
forms an admissible orthogonal refinement of \(R\). If
\[
w_1 = 0.40,\qquad
w_2 = 0.30,\qquad
w_3 = 0.20,\qquad
w_4 = 0.10,
\]
then the realized quadratic shares are exactly the subset sums
\[
0,\ 0.10,\ 0.20,\ 0.30,\ 0.40,\ 0.50,\ 0.60,\ 0.70,\ 0.80,\ 0.90,\ 1.
\]
This set is not dense in \([0,1]\), but it makes the combinatorial
mechanism explicit: admissible binary groupings of finer internal
record components generate realized norm splits through subset sums of
their projected weights.

If one passes instead to a sequence of finer admissible decompositions
for which the largest projected weight tends to zero, then
\cref{prop:local-dense-saturation-criterion} implies that the realized
quadratic shares become dense in \([0,1]\), and therefore the realized
norm values become dense in \([0,\norm{\Proj{R}\Psi}]\) as well.

This is only an illustration of the mechanism. It does not certify
that any concrete physical apparatus realizes the required admissible
refinement structure. Its role is narrower: it shows explicitly how
dense admissible saturation can arise from sufficiently fine robust
internal sub-record structure once admissible binary groupings are
available.
\end{remark}

\subsection{A toy structural example: two-outcome spin system}
\label{subsec:spin-example}

The following example is purely structural. It does not derive a
pointer basis from decoherence, and it does not claim to settle any
interpretive question about quantum measurement. Its sole purpose is
to verify that the definitions of \cref{sec:framework} are non-empty
and that the theorem delivers the expected result in the simplest
possible case.

\begin{remark}[Toy structural example]
\label{rem:spin-example}
Let $\HH = \mathbb{C}^2$ and let
\[
  \Psi = \alpha\,|\uparrow\rangle + \beta\,|\downarrow\rangle,
  \qquad
  |\alpha|^2 + |\beta|^2 = 1.
\]
Define two robust record sectors by
\[
  R_\uparrow := \mathrm{span}\{|\uparrow\rangle\},
  \qquad
  R_\downarrow := \mathrm{span}\{|\downarrow\rangle\}.
\]

\emph{Verification of \cref{def:robust-record-sector}.}
Each of $R_\uparrow$ and $R_\downarrow$ is a one-dimensional closed
subspace of $\HH$. They represent mutually exclusive, internally
readable record alternatives: no admissible continuation can realize
both simultaneously. Each subspace is stable under small
representational perturbations in the sense that its one-dimensional
span is not disrupted by admissible micro-recodings. The pair
$\{R_\uparrow, R_\downarrow\}$ forms an admissible exclusivity
structure in which $R_\uparrow \oplus R_\downarrow = \HH$.

\emph{Admissible orthogonal refinement.}
The decomposition $\HH = R_\uparrow \oplus R_\downarrow$ is an
admissible orthogonal refinement of the full next-step sector relative
to $\Psi$: every admissible continuation falls under exactly one of
$R_\uparrow$ or $R_\downarrow$, and no continuation remains outside
this partition. The projected components are
\[
  \phi_{R_\uparrow} = \Proj{R_\uparrow}\Psi = \alpha\,|\uparrow\rangle,
  \qquad
  \phi_{R_\downarrow} = \Proj{R_\downarrow}\Psi = \beta\,|\downarrow\rangle,
\]
with $\norm{\phi_{R_\uparrow}}^2 = |\alpha|^2$ and
$\norm{\phi_{R_\downarrow}}^2 = |\beta|^2$.

\emph{Binary saturation.}
In this two-dimensional setting the decomposition
$\HH = R_\uparrow \oplus R_\downarrow$ is the only non-trivial
admissible binary refinement singled out by the fixed record basis.
The example is therefore not meant to establish admissible binary
saturation in general. Its purpose is only to illustrate the minimal
binary setting in which the theorem applies once admissible binary
saturation is assumed for the sector under consideration.

\emph{Conclusion from \cref{thm:quadratic-uniqueness,cor:normalized-quadratic-weight}.}
By \cref{thm:quadratic-uniqueness}, the only non-negative
refinement-stable induced weight takes the form
$\Wpsi{R} = c\,\norm{\Proj{R}\Psi}^2$.
Since $R_\uparrow$ and $R_\downarrow$ form a complete orthogonal
decomposition with $|\alpha|^2 + |\beta|^2 = 1$,
\cref{cor:normalized-quadratic-weight} gives $c = 1$. Therefore
\[
  \Wpsi{R_\uparrow} = |\alpha|^2,
  \qquad
  \Wpsi{R_\downarrow} = |\beta|^2.
\]
This is the expected Born assignment. The example shows that the
structural apparatus of the paper is non-vacuous: the definitions are
satisfiable, the conditions are checkable, and the theorem delivers
the correct quadratic assignment in this minimal case.
\end{remark}

\subsection{Transition to the literature comparison}

The scope of the theorem can now be stated succinctly. Within an admissible Hilbert record layer, and under explicit structural conditions tied to internal equivalence and sufficient admissible refinement richness, the quadratic assignment is the only non-negative refinement-stable induced weight on robust record sectors.

The next section compares this route with several established families of Born-rule arguments and identifies the precise logical point at which the present approach differs from them.

%% file: sections/07_relation_to_existing_routes.tex
\section{Relation to Existing Routes}
\label{sec:existing-routes}

The present theorem belongs to a familiar landscape of attempts to account for the quadratic quantum weight, but its logical target and its carrier structure differ from the standard routes. The comparison is therefore best made not only at the level of conclusion, but at the level of what each route takes as primary, where its additive or constraining structure is placed, and what each route is trying to prove.

\subsection{Measure-theoretic uniqueness routes}

A natural comparison point is the family of Gleason-type results. Gleason's theorem shows that, on Hilbert spaces of dimension at least three, sufficiently well-behaved additive measures on the lattice of projections are represented by density operators \cite{Gleason1957}. Busch's extension shifts the setting from projection-valued assignments to generalized observables and thereby removes the qubit restriction \cite{Busch2003}.

These results are mathematically powerful, but they begin with a global measure-theoretic target. Additivity is built into the starting point as a property of the measure on projections or effects. The present paper does not replace that route, and it does not attempt to subsume it. It proves a different uniqueness theorem with a different theorem target and a different additive carrier. Its target is narrower, namely induced weights on robust record sectors only. Its additive primitive is not a measure on the full projector lattice, but finite additivity on disjoint continuation bundles through an extensive bundle valuation. The sector-level additive law is then inherited from continuation partition under admissible refinement. The difference therefore runs through the object being weighted, the structure that carries additivity, and the level at which the uniqueness question is posed.

This difference in additive carrier was already built into the framework level in \cref{sec:framework} and is here restated at the level of theorem comparison.

A recent critical analysis by Zhang argues, from a different angle, that additivity is indispensable across several major Born-rule derivations and cannot be eliminated in favour of weaker assumptions such as non-contextuality and normalization alone \cite{Zhang2026Additivity}. That diagnosis is broadly compatible with the present result, but the present paper sharpens the structural point in a different way: it does not treat additivity as a global measure-theoretic postulate on the projector lattice, but relocates the additive primitive to disjoint continuation bundles and then identifies the further structural conditions under which quadratic weight is forced on robust record sectors.

\subsection{Decision-theoretic Everettian routes}

A second major route is the Everettian decision-theoretic program initiated by Deutsch and developed in more formal detail by Wallace \cite{Deutsch1999,Wallace2010,Wallace2012}. In that framework, the Born rule is connected to constraints on rational preference or rational betting behaviour in branching quantum situations. Critical discussion of that route has emphasized both its formal ingenuity and the non-trivial role played by its rationality assumptions \cite{Kent2010}.

The present theorem does not appeal to rational choice, utility, diachronic consistency, or betting behaviour. It is not a theorem about how agents \emph{should} apportion credence in a branching setting. It is a theorem about which non-negative weight assignments are structurally admissible on robust record sectors under explicit refinement-stability conditions. The contrast is therefore not only between normative and non-normative language, but between two different theorem targets: rational credence constraints on the one hand, and refinement-stable induced weight on the other.

\subsection{Envariance and symmetry-based routes}

A third route is Zurek's envariance program, which seeks to derive the Born rule from symmetry properties of entangled states \cite{Zurek2005}. That approach has been influential precisely because it attempts to replace probabilistic postulates with a more primitive symmetry argument. At the same time, critical analyses have argued that substantial assumptions remain active in the derivation, especially in the passage from symmetric situations to probability assignments \cite{SchlosshauerFine2005}.

The present route does not rely on entanglement symmetry, swap invariance, or environment-assisted invariance. The core functional equation arises instead from two different ingredients: sector-level continuation partition under admissible refinement, and the reduction of induced weight to the norm profile of projected record components. The active invariance principle here is therefore not a symmetry of entangled states, but invariance of induced weight under internally indistinguishable admissible refinement structure.

\subsection{Self-locating uncertainty and evidential routes}

Another line of work approaches the Born rule through self-locating uncertainty and closely related questions about confirmation in Everettian settings \cite{SebensCarroll2014,GreavesMyrvold2010}. These approaches ask how an observer should reason when branching has occurred or is about to occur, and how observed frequencies can count as evidence in a theory with multiple successor observers.

The present theorem is logically prior to such questions. It does not begin with post-branch credence, self-location, or theory confirmation. Its concern is earlier and more structural: which weight assignments are compatible with internally stable record refinement in the first place. On the present route, evidential and self-locating questions arise only after the admissible weight form has already been fixed.

\subsection{Operational reconstruction routes}

A further comparison is with operational reconstructions in which the measurement postulates are argued to be fixed by the rest of quantum structure together with additional operational principles. A notable recent example is the claim that the measurement postulates of quantum mechanics are operationally redundant \cite{MasanesGalleyMuller2019}. That line differs from the present one both in scope and in target: it addresses the structure of measurements and state update more broadly, whereas the present paper isolates only the uniqueness of induced weight on robust record sectors. The present result is therefore not a reconstruction of measurement theory, but a structural uniqueness theorem for record-sector weight once the admissible framework is fixed. The subsequent critical discussion of hidden assumptions in such reconstructions also underscores the importance of making the active structural conditions fully explicit \cite{Stacey2022}.

\subsection{Normative Bayesian routes}

Finally, there are approaches, most prominently in QBism, that do not treat the Born rule as an objective branch-weight law at all, but as a normative coherence condition governing an agent's probability assignments \cite{FuchsSchack2013}. On that reading, the Born rule is not derived from the structure of branching worlds or record sectors, but interpreted as an empirically constrained rule for consistent probabilistic judgement.

The present theorem is different in both aim and object. It does not propose a normative rule for an agent's betting commitments. It identifies the only non-negative refinement-stable induced weight on robust record sectors under explicitly stated conditions on admissible refinement and internal indistinguishability.

\subsection{Logical difference of the present route}

The clearest way to summarize the difference is to ask where each route places its primary structure. Standard routes typically begin with one of the following:
\begin{enumerate}[label=(\roman*)]
\item a global measure on projections or effects,
\item rational preference in branching settings,
\item symmetry of entangled states,
\item self-locating credence after branching,
\item broad operational postulates about measurement,
\item normative coherence constraints on probability assignments.
\end{enumerate}

The present paper begins elsewhere. It starts from robust record sectors, admissible continuation bundles, an extensive bundle valuation on disjoint bundles, and explicit structural conditions on admissible refinement. Its theorem is correspondingly narrower than a full reconstruction of quantum probability, but sharper in logical target. What is shown is not that every projector carries Born weight by abstract measure representation, nor that every rational Everettian agent should bet in accordance with amplitude-squared weight, but that within an admissible Hilbert record layer the quadratic assignment is the only non-negative refinement-stable induced weight on robust record sectors.

The present route should therefore not be described as a replacement for Gleason, but as a distinct conditional uniqueness theorem at a different structural level.

%% file: sections/08_conclusion.tex
\section{Conclusion}
\label{sec:conclusion}

This paper has identified a conditional structural uniqueness theorem for quadratic weight that differs in target, additive carrier, and logical structure from the standard derivations. The argument does not begin with a probability measure on the full projector lattice, with rational betting constraints, with entanglement symmetries, or with self-locating credence. It begins with robust record sectors, admissible continuation bundles, and explicit structural conditions on admissible refinement.

Within an admissible Hilbert record layer, the argument proceeds in
three steps. First, admissible orthogonal refinement induces a
partition of continuation bundles and therefore an additive law for
sector-level induced weight. Second, internal equivalence together
with the norm classification supplied by admissible binary saturation
reduces the sector-level assignment to a one-variable function of the
norm of the projected record component. Third, sufficient admissible refinement
richness forces the resulting functional equation: in the main theorem
this is secured by admissible binary saturation, while
\cref{prop:dense-saturation} shows that dense admissible saturation
already suffices if continuity of the profile function is added.
Non-negativity then forces the additive function to be linear by
\cref{lem:cauchy}. The result is the quadratic assignment
\[
\Wpsi{R} = c\,\norm{\Proj{R}\Psi}^2,
\]
with the normalized case giving the standard Born assignment
\[
\Wpsi{R} = \norm{\Proj{R}\Psi}^2.
\]

The main conceptual point is therefore not merely that a quadratic
expression reappears. It is that, on the stated domain, the quadratic
assignment is the only non-negative refinement-stable induced weight.
The theorem isolates a specific structural threshold: once the
additive primitive is placed on continuation bundles, binary-profile-equivalent
pairs carry the same induced weight, and the admissible refinement
structure is rich enough to force the functional equation, no other
non-negative induced weight can be refinement-stable.
This source of uniqueness sits at a different structural level from the standard global measure-theoretic and decision-theoretic routes. The framework is not merely a convenient domain restriction. It isolates a functional setting in which induced weight on record-like alternatives is both non-arbitrary and structurally constrainable.

Just as important is the restricted scope of the result. The theorem neither derives Hilbert structure nor assigns weights to arbitrary projectors, and it does not claim that every orthogonal decomposition is physically meaningful or that broader metaphysical conclusions about branching structure follow by themselves. What it does show is narrower and sharper: once one works within an admissible Hilbert record layer and imposes the stated structural conditions, the quadratic assignment is forced as the only non-negative refinement-stable induced weight on robust record sectors.

That conditional form should be regarded as a strength. It makes the
threshold structure of the argument transparent and keeps the
uniqueness claim tied to explicitly stated premises. It also separates
two questions that are often blurred together: the mathematical
uniqueness question and the physical applicability question. Future
work can therefore proceed in a controlled way. One direction is to
investigate whether the internal equivalence principle can be justified
from a still deeper characterization of internally accessible record
structure. Another is to study which physically realistic record
sectors satisfy admissible binary saturation, or at least dense
admissible saturation together with continuity, and where these
conditions break down. A third is to ask whether the present route can
be extended beyond binary refinements while preserving the same degree
of explicitness.

The present result should therefore be read as a focused theorem about
induced weight on robust record sectors. Its claim is not maximal
generality, but explicit structural uniqueness at a clearly identified
threshold.
Under the stated conditions, the quadratic assignment is
not one admissible choice among many. It is the only non-negative
refinement-stable one, and the structural conditions that force this
conclusion have been stated explicitly enough to be evaluated,
contested, refined, or further mapped at the structural and empirical
level.

%% file: bib/references.bib
@article{Gleason1957,
  author  = {Gleason, Andrew M.},
  title   = {Measures on the Closed Subspaces of a {Hilbert} Space},
  journal = {Journal of Mathematics and Mechanics},
  volume  = {6},
  number  = {4},
  pages   = {885--893},
  year    = {1957},
  doi     = {10.1512/iumj.1957.6.56050}
}

@article{Busch2003,
  author  = {Busch, Paul},
  title   = {Quantum States and Generalized Observables: A Simple Proof of {Gleason's} Theorem},
  journal = {Physical Review Letters},
  volume  = {91},
  number  = {12},
  pages   = {120403},
  year    = {2003},
  doi     = {10.1103/PhysRevLett.91.120403}
}

@article{Deutsch1999,
  author  = {Deutsch, David},
  title   = {Quantum Theory of Probability and Decisions},
  journal = {Proceedings of the Royal Society A: Mathematical, Physical and Engineering Sciences},
  volume  = {455},
  number  = {1988},
  pages   = {3129--3137},
  year    = {1999},
  doi     = {10.1098/rspa.1999.0443}
}

@incollection{Wallace2010,
  author    = {Wallace, David},
  title     = {A Formal Proof of the {Born} Rule from Decision-Theoretic Assumptions},
  booktitle = {Many Worlds?: {Everett}, Quantum Theory, and Reality},
  editor    = {Saunders, Simon and Barrett, Jonathan and Kent, Adrian and Wallace, David},
  publisher = {Oxford University Press},
  address   = {Oxford},
  pages     = {227--263},
  year      = {2010},
  doi       = {10.1093/acprof:oso/9780199560561.003.0010}
}

@book{Wallace2012,
  author    = {Wallace, David},
  title     = {The Emergent Multiverse: Quantum Theory According to the {Everett} Interpretation},
  publisher = {Oxford University Press},
  address   = {Oxford},
  year      = {2012},
  doi       = {10.1093/acprof:oso/9780199546961.001.0001}
}

@article{Zurek2005,
  author  = {Zurek, Wojciech H.},
  title   = {Probabilities from Entanglement, {Born's} Rule from Envariance},
  journal = {Physical Review A},
  volume  = {71},
  number  = {5},
  pages   = {052105},
  year    = {2005},
  doi     = {10.1103/PhysRevA.71.052105}
}

@article{SchlosshauerFine2005,
  author  = {Schlosshauer, Maximilian and Fine, Arthur},
  title   = {On {Zurek's} Derivation of the {Born} Rule},
  journal = {Foundations of Physics},
  volume  = {35},
  number  = {2},
  pages   = {197--213},
  year    = {2005},
  doi     = {10.1007/s10701-004-1941-6}
}

@incollection{SebensCarroll2014,
  author    = {Sebens, Charles T. and Carroll, Sean M.},
  title     = {Many Worlds, the {Born} Rule, and Self-Locating Uncertainty},
  booktitle = {Quantum Theory: A Two-Time Success Story},
  editor    = {Struppa, Daniele C. and Tollaksen, Jeffrey M.},
  publisher = {Springer},
  address   = {Milan},
  pages     = {157--169},
  year      = {2014},
  doi       = {10.1007/978-88-470-5217-8_10},
  note      = {Preprint version: arXiv:1405.7907 [quant-ph]}
}

@incollection{GreavesMyrvold2010,
  author    = {Greaves, Hilary and Myrvold, Wayne},
  title     = {{Everett} and Evidence},
  booktitle = {Many Worlds?: {Everett}, Quantum Theory, and Reality},
  editor    = {Saunders, Simon and Barrett, Jonathan and Kent, Adrian and Wallace, David},
  publisher = {Oxford University Press},
  address   = {Oxford},
  pages     = {264--304},
  year      = {2010},
  doi       = {10.1093/acprof:oso/9780199560561.003.0011}
}

@incollection{Kent2010,
  author    = {Kent, Adrian},
  title     = {One World Versus Many: The Inadequacy of {Everettian} Accounts of Evolution, Probability, and Scientific Confirmation},
  booktitle = {Many Worlds?: {Everett}, Quantum Theory, and Reality},
  editor    = {Saunders, Simon and Barrett, Jonathan and Kent, Adrian and Wallace, David},
  publisher = {Oxford University Press},
  address   = {Oxford},
  pages     = {307--354},
  year      = {2010},
  doi       = {10.1093/acprof:oso/9780199560561.003.0012}
}

@article{MasanesGalleyMuller2019,
  author  = {Masanes, Llu{\'\i}s and Galley, Thomas D. and M{\"u}ller, Markus P.},
  title   = {The Measurement Postulates of Quantum Mechanics Are Operationally Redundant},
  journal = {Nature Communications},
  volume  = {10},
  pages   = {1361},
  year    = {2019},
  doi     = {10.1038/s41467-019-09348-x}
}

@misc{Stacey2022,
  author        = {Stacey, Blake C.},
  title         = {{Masanes--Galley--M{\"u}ller} and the State-Update Postulate},
  year          = {2022},
  eprint        = {2211.03299},
  archivePrefix = {arXiv},
  primaryClass  = {quant-ph},
  doi           = {10.48550/arXiv.2211.03299},
  note          = {Preprint}
}

@article{FuchsSchack2013,
  author  = {Fuchs, Christopher A. and Schack, R{\"u}diger},
  title   = {{Quantum-Bayesian} Coherence},
  journal = {Reviews of Modern Physics},
  volume  = {85},
  number  = {4},
  pages   = {1693--1715},
  year    = {2013},
  doi     = {10.1103/RevModPhys.85.1693}
}

@article{Caves2002,
  author  = {Caves, Carlton M. and Fuchs, Christopher A. and Manne, Kiran K. and Renes, Joseph M.},
  title   = {{Gleason}-Type Derivations of the Quantum Probability Rule for Generalized Measurements},
  journal = {Foundations of Physics},
  volume  = {34},
  number  = {2},
  pages   = {193--209},
  year    = {2004},
  doi     = {10.1023/B:FOOP.0000019581.00318.a5}
}

@article{Zurek2003,
  author  = {Zurek, Wojciech H.},
  title   = {Environment-Assisted Invariance, Entanglement, and Probabilities in Quantum Physics},
  journal = {Physical Review Letters},
  volume  = {90},
  number  = {12},
  pages   = {120404},
  year    = {2003},
  doi     = {10.1103/PhysRevLett.90.120404}
}

@misc{Zhang2026Additivity,
  author        = {Zhang, Jiaxuan},
  title         = {Summing to Uncertainty: On the Necessity of Additivity in Deriving the {Born} Rule},
  year          = {2026},
  eprint        = {2603.06211},
  archivePrefix = {arXiv},
  primaryClass  = {quant-ph},
  doi           = {10.48550/arXiv.2603.06211},
  note          = {Preprint}
}

@article{Zurek2003Decoherence,
  author  = {Zurek, Wojciech H.},
  title   = {Decoherence, Einselection, and the Quantum Origins of the Classical},
  journal = {Reviews of Modern Physics},
  volume  = {75},
  number  = {3},
  pages   = {715--775},
  year    = {2003},
  doi     = {10.1103/RevModPhys.75.715}
}

@book{Schlosshauer2007Decoherence,
  author    = {Schlosshauer, Maximilian},
  title     = {Decoherence and the Quantum-to-Classical Transition},
  publisher = {Springer},
  address   = {Berlin},
  series    = {The Frontiers Collection},
  year      = {2007},
  doi       = {10.1007/978-3-540-35775-9}
}

@article{OllivierPoulinZurek2004,
  author  = {Ollivier, Harold and Poulin, David and Zurek, Wojciech H.},
  title   = {Objective Properties from Subjective Quantum States: Environment as a Witness},
  journal = {Physical Review Letters},
  volume  = {93},
  number  = {22},
  pages   = {220401},
  year    = {2004},
  doi     = {10.1103/PhysRevLett.93.220401}
}

@article{Zurek2009QuantumDarwinism,
  author  = {Zurek, Wojciech H.},
  title   = {Quantum Darwinism},
  journal = {Nature Physics},
  volume  = {5},
  pages   = {181--188},
  year    = {2009},
  doi     = {10.1038/nphys1202}
}
